\documentclass[sigconf]{acmart}
\acmConference[ICSE 2024]{46th International Conference on Software Engineering}{April 2024}{Lisbon, Portugal}

\AtBeginDocument{%
  \providecommand\BibTeX{{%
    \normalfont B\kern-0.5em{\scshape i\kern-0.25em b}\kern-0.8em\TeX}}}

\copyrightyear{2024}
\acmYear{2024}
\setcopyright{rightsretained}
\acmConference[ICSE '24]{2024 IEEE/ACM 46th International Conference on Software Engineering}{April 14--20, 2024}{Lisbon, Portugal}
\acmBooktitle{2024 IEEE/ACM 46th International Conference on Software Engineering (ICSE '24), April 14--20, 2024, Lisbon, Portugal}\acmDOI{10.1145/3597503.3639151}
\acmISBN{979-8-4007-0217-4/24/04}

\usepackage{algorithmic}
\usepackage{graphicx}
\usepackage{textcomp}
\usepackage{xcolor}

\newtheorem{definition}{Definition}
\usepackage{todonotes}
\usepackage{comment}

\usepackage[normalem]{ulem}

\usepackage[ruled,linesnumbered]{algorithm2e}
\usepackage{hyperref}

\usepackage{multirow}
\usepackage{tabularx}
\usepackage{enumitem}
\usepackage{caption}
\usepackage{verbatim}
\usepackage{diagbox}

\usepackage{makecell}
\usepackage{balance}

\usepackage{listings}
\usepackage{fancyvrb}

\usepackage{tikz}
\usepackage{tikz-qtree}

\newtheorem{theorem}{Theorem}[section]

\newtheorem{proposition}[theorem]{Proposition}

\newtheorem{example}[theorem]{Example}

\newcommand{\coolname}{$\mathtt{REDriver}$\xspace}

\begin{document}

\title{REDriver: Runtime Enforcement for Autonomous Vehicles}

\author{Yang Sun}
\orcid{0000-0002-2409-2160}
 \affiliation{%
   \institution{Singapore Management University}
   \country{Singapore}
}
\email{yangsun.2020@phdcs.smu.edu.sg}

\author{Christopher M. Poskitt}
\orcid{0000-0002-9376-2471}
\affiliation{\institution{Singapore Management University}\country{Singapore}}
\email{cposkitt@smu.edu.sg}

 \author{Xiaodong Zhang}
 \orcid{0000-0002-8380-1019}
 \affiliation{%
   \institution{Xidian University}
   \country{China}
   }
 \email{zhangxiaodong@xidian.edu.cn}

\author{Jun Sun}
\orcid{0000-0002-3545-1392}
\affiliation{\institution{Singapore Management University}\country{Singapore}}
\email{junsun@smu.edu.sg}

\begin{abstract}
    Autonomous driving systems (ADSs) integrate sensing, perception, drive control, and several other critical tasks in autonomous vehicles, motivating research into techniques for assessing their safety.
    While there are several approaches for testing and analysing them in high-fidelity simulators, ADSs may still encounter additional critical scenarios beyond those covered once they are deployed on real roads.
    An additional level of confidence can be established by monitoring and enforcing critical properties when the ADS is running.
    Existing work, however, is only able to monitor simple safety properties (e.g., avoidance of collisions) and is limited to blunt enforcement mechanisms such as hitting the emergency brakes.
    In this work, we propose \coolname, a general and modular approach to runtime enforcement, in which users can specify a broad range of properties (e.g., national traffic laws) in a specification language based on signal temporal logic~(STL).
    \coolname monitors the planned trajectory of the ADS based on a quantitative semantics of STL, and uses a gradient-driven algorithm to repair the trajectory when a violation of the specification is likely.
    We implemented \coolname for two versions of Apollo (i.e., a popular ADS), and subjected it to a benchmark of violations of Chinese traffic laws.
    The results show that \coolname significantly improves Apollo's conformance to the specification with minimal overhead.
\end{abstract}

\maketitle

\section{Introduction}

Autonomous driving systems~(ADSs) are the core of autonomous vehicles~(AVs), integrating sensing, perception, drive control, and several other tasks that are necessary for automating their journeys.
Given the safety-critical nature of ADSs \cite{favaro2017examining,dixit2016autonomous}, it is imperative that they operate safely at all times, including in rare or unexpected scenarios that may not have been explicitly considered when the system was designed.
This has spurred a multitude of research into techniques for establishing confidence in an ADS, e.g., by modelling and verifying aspects of its design~\cite{Gu-et_al19a}, by subjecting it to reconstructions of real-world accidents~\cite{Bashetty20DeepCrashTest}, or by testing it against automatically generated critical scenarios~\cite{li2020av,Sun-Poskitt-et_al22a,Zhou-et_al23a} in a high-fidelity simulator such as CARLA~\cite{dosovitskiy2017carla} or LGSVL~\cite{rong2020lgsvl}.

These approaches all analyse an ADS \emph{before} it is deployed on real roads, where it may still encounter additional scenarios beyond those that were covered. In fact, an analysis of accidents involving autonomous vehicles~\cite{mccarthy2022autonomous} suggests that the broader implementation of current AV technologies may not lead to a reduction in vehicle crash frequency.
An additional level of confidence can thus be established if desirable properties are also monitored---even enforced---while the ADS is running.
This is the idea of \emph{runtime enforcement}, a technique that observes the execution of a system and then modifies it in a minimal way to ensure certain properties are satisfied.
In AVs, runtime enforcement has been applied, for example, to monitor basic safety properties such as the avoidance of collisions, applying the emergency brake before they are violated~\cite{Grieser-et_al20a}.
Avoiding collisions, however, is not enough in general.
ADSs are expected to satisfy a broader range of complicated properties concerning the overall traffic systems they operate in, such as national traffic laws that describe how vehicles should behave with respect to various junctions, signals, and (most precariously) other vehicles or pedestrians.
Currently, no existing approach supports runtime enforcement of properties in this direction.

In this work, we aim to provide a general solution to the runtime enforcement problem for AVs.
In particular, we propose \coolname, a general framework for runtime enforcement that can be integrated into ADSs with state-of-the-art modular designs, as exhibited by Apollo~\cite{apollo70} and Autoware~\cite{autoware}.
\coolname allows users to specify desirable properties of AVs using an existing and powerful domain-specific language~(DSL) based on signal temporal logic~(STL).
This language supports properties ranging from the simplest, concerning collision avoidance, through to entire formalisations of national traffic laws~\cite{Sun-Poskitt-et_al22a}.
\coolname monitors the planned trajectories and command sequences of the ADS at runtime and assesses them against the user's specifications.
If the AV is predicted to potentially violate them in the near future (based on a quantitative semantics of STL), \coolname repairs the trajectories using a gradient-driven algorithm.
Furthermore, it does so while minimising the ``overhead'' (or change) to the original journey. That is, by efficiently computing the gradient of each signal (with respect to the robustness degree of the STL formula), we identify and modify the signal that is most likely to repair the trajectories.

\begin{figure}[t]
    \centering
    \includegraphics[width=1\linewidth]{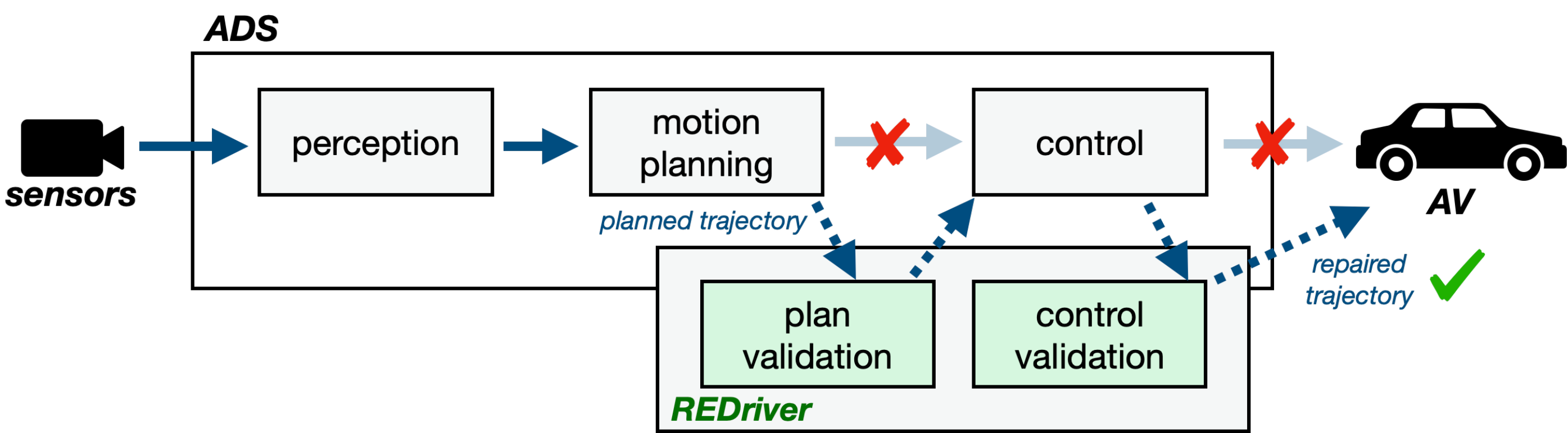}
    \caption{The architecture of an ADS with \coolname}
    \label{fig:workflow}
\end{figure}

\coolname has been implemented for two versions of Apollo (i.e., versions 6.0 and 7.0, the latest at the time of experimentation).
The implementation consists of a \emph{plan validation} algorithm and a \emph{control validation} algorithm that respectively observe and modify (if necessary) the outputs of the ADSs' motion planning and control modules. Note that the motion planning and control modules are black boxes to us.
In particular, we enforce that these outputs (i.e., planned trajectories and command sequences) do not lead to violations---whenever possible---of a comprehensive formalisation of the Chinese traffic laws.
This goes far beyond existing runtime enforcement approaches, which focus on simple safety properties (e.g., collision avoidance) and blunt enforcement mechanisms (e.g., hitting the emergency brakes).
Figure~\ref{fig:workflow} depicts how \coolname is integrated into the modular design of Apollo. In particular, we have added two new modules while ensuring that the existing modules and their inner logic remain unchanged. In the diagram, the perception, motion planning, and control boxes represent the existing Apollo modules, while the green plan validation and control validation boxes represent the new modules from \coolname. The arrow denotes the flow of signal transmission.
We evaluated our implementation of \coolname against a benchmark of violation-inducing scenarios for Chinese traffic laws~\cite{Sun-Poskitt-et_al22a}, finding that our runtime enforcement approach significantly reduces the likelihood of those violations occurring.
Furthermore, \coolname's overhead in terms of time and how often it intervenes is negligible.

\section{Background and Problem}
\label{sec:overview_and_background}

In this section, we review the architecture of ADSs, the DSL for specifying safety properties, and then define our problem.

\subsection{Overview of Autonomous Driving Systems}
State-of-the-art open-source ADSs such as Apollo~\cite{apollo60} and Autoware~\cite{autoware} have similar architectures.
They are typically organised into loosely coupled modules that communicate via message-passing.
Three of these modules are particularly relevant to our context, i.e., perception, motion planning, and control.

\begin{table}[!t]
\setlength{\abovecaptionskip}{0.cm}
	\centering
	\caption{An example planned trajectory}
	\label{tab:example_of_the_planned_trajectory}
    \begin{tabular}{|c|c|c|c|c|c|} \hline
         Time & Position & Speed & Acc & Steer & Gear  \\
         \hline
         0 & (x: 0, y: 0)           & 7.01 & -0.05 & 0 & $\mathtt{DRIVE}$\\     
         2 & (x: 0, y: 13.34)       & 6.13 & -0.48 & 0 & $\mathtt{DRIVE}$\\     
         4 & (x: 0, y: 24.83)       & 5.44 & -0.24 & 0 & $\mathtt{DRIVE}$\\     
         6 & (x: 0, y: 35.85)       & 5.09 & -0.18 & 0 & $\mathtt{DRIVE}$\\     
         8 & (x: 0, y: 44.75)       & 3.89 & -1.44 & 0 & $\mathtt{DRIVE}$\\   \hline   
         

    \end{tabular}
\vspace{-0.5cm}
\end{table}

\begin{table}[!t]
\setlength{\abovecaptionskip}{0.cm}
	\centering
	\caption{An example predicted environment}
	\label{tab:example_of_the_predicted_environment}
    \begin{tabular}{|c|c|c|c|c|c|c|} \hline
         Type & Time & Position & Speed & Acc & Steer  \\
         \hline
         \multirow{4}{*}{Car1}  & 0 & (x: 2.5, y: 5)            & 7.42 & -0.05 & -7.25 \\ 
                                & 2 & (x: 1.67, y: 18.34)       & 6.37 & -0.48 & -11.10 \\ 
                                & 4 & (x: 0, y: 29.88)          & 5.44 & -0.24 & 0 \\ 
                                & 6 & (x: 0, y: 40.87)          & 5.09 & -0.18 & 0 \\ 
                                & 8 & (x: 0, y: 49.76)          & 3.89 & -1.44 & 0 \\ \hline   
        \multirow{4}{*}{Car2}   & 0 & (x: -2.5, y: 15)      & 0 & 0 & 0 \\ 
                                & 2 & (x: -2.5, y: 15)      & 0 & 0 & 0 \\ 
                                & 4 & (x: -2.5, y: 15)      & 0 & 0 & 0 \\ 
                                & 6 & (x: -2.5, y: 15)      & 0 & 0 & 0 \\ 
                                & 8 & (x: -2.5, y: 15)      & 0 & 0 & 0 \\ \hline   
        \multirow{4}{*}{Ped1}   & 0 & (x: 0.23, y: 48)      & 0 & 0 & 0 \\ 
                                & 2 & (x: 0.23, y: 48)      & 0 & 0 & 0 \\ 
                                & 4 & (x: 0.23, y: 48)      & 0 & 0 & 0 \\ 
                                & 6 & (x: 0.23, y: 48)      & 0 & 0 & 0 \\ 
                                & 8 & (x: 0.23, y: 48)      & 0 & 0 & 0 \\ \hline   
        TL-ID & Time & Color & Blink & -- & -- \\
        \hline
         \multirow{4}{*}{TL-0}  & 0 & GREEN     & False & -- & --\\ 
                                & 2 & YELLOW    & False & -- & --\\ 
                                & 4 & YELLOW    & False & -- & --\\ 
                                & 6 & YELLOW    & False & -- & --\\ 
                                & 8 & RED       & False & -- & --\\ \hline   
    \end{tabular}
\vspace{-0.5cm}
\end{table}

First, the perception module receives sensor readings (e.g., from a camera or LIDAR), processes them, and then feeds them to the motion planning module.
Second, the motion planning module generates a \emph{planned trajectory} based on the map, the destination, the sensor inputs, and the state of the ego vehicle, i.e., the one under the control of the ADS.
Intuitively, the planned trajectory describes where the vehicle will be at future time points, and is computed based on a predicted environment that includes, for example, the predicted trajectories of other vehicles (NPCs, non-player characters), pedestrians, and traffic lights.
For instance, Table~\ref{tab:example_of_the_planned_trajectory} shows a planned trajectory for an ego vehicle with respect to the predicted environment shown in Table~\ref{tab:example_of_the_predicted_environment}.
Here, the ego vehicle slows down before approaching an intersection as the traffic light is changing to red.
Every line in Table~\ref{tab:example_of_the_planned_trajectory} represents a planned \emph{waypoint}, i.e., the planned position, speed, acceleration, steer, and gear of the ego vehicle at a series of future time points. Note that an actual planned trajectory typically contains hundreds of waypoints. 
Similarly, every line in Table~\ref{tab:example_of_the_predicted_environment} corresponds to the predicted states of NPCs such as vehicles and pedestrians, as well as environmental parameters like traffic lights. Here, Car2 and Ped1 are predicted to be stationary, Car1 is predicated to change lanes, and the color of the traffic light ahead is predicted to change from green to yellow and eventually to red.
Furthermore, in general there may be multiple planned trajectories for a given destination, and the planning module attempts to find the ``best'' one.
Finally, the control module translates the planned trajectory into control commands (e.g., `brake', and `signal') so that the ego vehicle is likely to follow the planned trajectory, i.e., passing through the waypoints with the planned speed, acceleration, steering angle, and gear position.
We refer to~\cite{apollo60,autoware} for details on how commands are generated.

There may be other modules in an ADS (e.g., the map module in Apollo) or the above-mentioned modules may be further divided into sub-modules (e.g., motion planning in Apollo is divided into routing, prediction, and planning).
Nonetheless, the similar high-level design of existing ADSs implies we could potentially introduce a module for runtime monitoring and enforcement which sits in-between existing modules, i.e., to intercept, analyse, and alter (if necessary) the inter-module messages. This way, runtime monitoring and enforcement can be introduced without changing the inner logic of existing modules. For instance, given the planned trajectory generated by the planning module shown in Table~\ref{tab:example_of_the_planned_trajectory}, if we decide that the trajectory could potentially lead to the violation of a certain property, we can simply modify the planned trajectory before forwarding it to the control module (to trigger a different control command generation).

\subsection{Property Specification}
\label{sec:Specification definition}
To go beyond the simplest safety requirements (e.g., `the ego vehicle does not collide'), we require a specification language that is able to express a rich set of properties that are relevant to autonomous vehicles and driving in general.
In this work, we adopt the driver-oriented specification language of LawBreaker~\cite{Sun-Poskitt-et_al22a}, which is based on signal temporal logic~(STL), and has been demonstrated to be expressive enough to specify the traffic laws of China and Singapore.
We highlight its key features, referring readers to~\cite{Sun-Poskitt-et_al22a} for details.

\begin{figure}[t]
    \centering
    \begin{align*}
    \varphi\;:=&\;\mu \;
     |\;\neg\varphi\;
     |\;\varphi_1 \lor \varphi_2\;|\; \varphi_1 \land \varphi_2\;|\; \varphi_1 \;\mathtt{U_I}\; \varphi_2\;\\ 
    \mu\;:=&f(x_0,x_1,\cdots, x_k) \sim 0 \ \ \sim \;:= > \;| \geq\;| <\;| \leq\;| \neq\;| =;
\end{align*}
    \caption{Specification language syntax, where $\varphi$, $\varphi_1$ and $\varphi_2$ are STL formulas, $I$ is an interval, and $f$ is a multivariate linear continuous function over language variables $x_i$}
    \label{syntax}
\vspace{-0.5cm}
\end{figure}

The high-level syntax of the language is shown in Figure~\ref{syntax}. A time interval $I$ is of the form $[l,u]$ where $l$ and $u$ are respectively the lower and upper bounds of the interval. 
Following convention, we write $\Diamond_I \; \varphi$ to denote $true \; \mathtt{U}_I \; \varphi$; and $\Box_I \; \varphi$ to denote $\neg \; \Diamond_I \; \neg \varphi$. Intuitively, $\mathtt{U}$, $\Box$, and $\Diamond$ are modal operators that are respectively interpreted as `until', `always', and `eventually'. Note that the time interval is omitted when it is $[0, \infty]$. 
The propositions in this language are constructed using 17 variables and 16 functions that are relevant to AVs, some of which are shown in Tables~\ref{tab:planning_trajectory_variables}.
In general, $\mu$ can be regarded as a proposition of the form $f(x_0,x_1,\cdots, x_k) \sim 0$ where $f$ is a multivariate linear continuous function and $x_i$ for all $i$ in $[0,k]$ is a variable supported in the language.

\begin{example}
\label{example:specifications}
Consider the following two (English translations of) traffic rules from the Regulations for Road Traffic Safety of the People's Republic of China~\cite{China_traffic_law}.
\begin{enumerate}
    \item \emph{Article \#38-(3): When a red light is on, vehicles are prohibited from passing.
    However, vehicles turning right can pass without hindering the passage of vehicles or pedestrians.}
    \item \emph{Article \#58-(3): When a vehicle is driving on a foggy day, the fog lights and hazard warning flashing should be on.}
\end{enumerate}
The above can be formalised as follows.
\begin{align*}
    law38_{3} & \equiv \Box ((TL(color) = red \\
     & \land (D(stopline) < 2 \lor D(junction) < 2) \\
     & \land \lnot direction = right) \to (\Diamond_{[0,3]}(speed < 0.5)) \\
     & \land (TL(color) = red  \land (D(stopline) < 2 \\
     & \lor D(junction) < 2)  \land direction = right \\
     & \land \lnot PriorityV(20)  \land \lnot PriorityP(20))\\
     & \to (\Diamond_{[0,2]}(speed > 0.5))) \\
     law58_{3} & \equiv \Box (fog \geq 0.5 \to (fogLight \land warningFlash)) 
\end{align*}
where $\mathtt{speed}$, $\mathtt{direction}$, $\mathtt{fogLight}$, and $\mathtt{warningFlash}$ represent the speed, direction, fog light status, and warning flash light status of the vehicle; $\mathtt{TL()}$ returns the status of traffic light ahead; $\mathtt{D(object)}$ calculates the distance from the vehicle to the $\mathtt{object}$ ahead; and $\mathtt{PriorityV(n)}$, $\mathtt{PriorityP(n)}$ check whether there is a priority vehicle or pedestrian within $\mathtt{n}$ meters ahead.
Note that several configurable constants (e.g., the distance $\mathtt{2}$ and the time interval $[0,3]$) are introduced to reduce the vagueness of the law in practice~\cite{Sun-Poskitt-et_al22a}. 
    \qed
\end{example}
A specification is evaluated with respect to a trace $\pi$ of \emph{scenes}, denoted as $\pi=\langle \pi_0, \pi_1, \pi_2 \ldots, \pi_n \rangle$, where each scene $\pi_i$ is a valuation of the propositions at time step $i$ and $\pi_0$ reflects the state at the start of a simulation.
These traces can be constructed from the planned trajectory generated by the ADS (Section~\ref{sec:constrct_trace}). 
We follow the standard semantics of STL (see e.g., \cite{maler2004monitoring}).

\begin{table}[t]
\setlength{\abovecaptionskip}{0.cm}
	\centering
	\caption{\small Car and environment related variables}
	\label{tab:planning_trajectory_variables}
	\footnotesize
    \begin{tabular}{c|c|p{1.55in}}
         \textbf{Signal} & \textbf{Type} & \textbf{Remarks} \\
         \hline
        $\mathtt{speed}$ & Num & Speed of ego vehicle ($m/s$). \\
         $\mathtt{acc}$ & Num & Acceleration of ego veh ($m/s^2$). \\
        $\mathtt{direction}$ & Enum & $\mathtt{forward}$, $\mathtt{left}$, $\mathtt{right}$ \\

         

        $\mathtt{D(stopline)}$ & Num & distance to the stopline ahead \\
        $\mathtt{D(junction)}$ & Num & distance to the junction ahead \\
         
        $\mathtt{fogLight}$ & Bool & whether the fog light is on \\
        $\mathtt{warningFlash}$ & Bool & whether the warning flash light is on \\


         

         $\mathtt{PriorityV(n)}$ & Bool & Whether there are vehicles with priority within n meters\\
         $\mathtt{PriorityP(n)}$ & Bool & Whether there are pedestrians with priority within n meters\\
         
         
         
        $\mathtt{TL(color)}$ & Enum & $\mathtt{YELLOW}$, $\mathtt{GREEN}$, $\mathtt{RED}$, or $\mathtt{BLACK}$ \\
         $\mathtt{TL(blink)}$ & Bool & if the traffic light ahead is blinking \\
         

        $\mathtt{fog}$ & Num &  degree of fog ranging from 0 to 1 \\
         $\mathtt{snow}$ & Num & degree of snow ranging from 0 to 1 \\
    \end{tabular}
\end{table}

\subsection{The Runtime Enforcement Problem for AVs}
Given an ADS and a user-specified property $\varphi$, our goal is to solve the \emph{runtime enforcement} problem for AVs by monitoring traces $\pi$ of the ADS against $\varphi$ at runtime, and altering its behavior when a violation is likely in the near future. Here, altering the ADS's behavior means adjusting its planned trajectory and consequently the control commands. Solving this problem could systematically improve the safety of ADSs when encountering unusual situations on the road.
We formulate our problem as follows:
\begin{definition}[Problem Definition]
\label{def:problem_def}
    Given a runtime planned trajectory $\gamma$, runtime control commands $\zeta$, a specification of ADS behavior $\varphi$, and a trace $\pi$ of the AV in a scenario.
    Let $\gamma'$, $\zeta'$, and $\pi'$ denote the adjusted planned trajectory, modified control commands, and resulting trace of the AV after these adjustments.
    Our problem is: 
    \begin{align*}
        Maximise: &~\frac{\rho(\varphi, \pi') - \rho(\varphi, \pi)}{|\gamma'- \gamma| + |\zeta'- \zeta|}. 
    \end{align*} \qed
\end{definition}
Intuitively, we seek to maximize the improvement in adhering to the desired behavior, while considering the magnitude of changes made to the planned trajectory and control commands. Here, the function $\rho$ serves as the quantitative semantics of a trace concerning the specification. Its purpose is to provide a numerical assessment that calculates the distance to a violation of the specification.

\section{Our Approach}
\label{sec:our_approach}
\label{sec:Specification Definition and Robustness Calculation}

\coolname, our runtime enforcement approach, consists of three broad steps.
First, \emph{plan validation}, in which it evaluates the planned trajectory against the specification to determine if there is a risk of violation.
Second, \emph{trajectory repair}, in which the planned trajectory is modified so as to avoid the violation.
Finally, \emph{control validation}, in which the commands generated by the control module are further evaluated to ensure the specification is satisfied.
As shown in Figure~\ref{fig:workflow}, these steps seamlessly integrate into the modular design of ADSs: \coolname sits between the modules, intercepting and altering the messages they exchange. Note that we assume that the sensor data received by the ADS is accurate.

\subsection{Plan Validation} \label{sec:constrct_trace}

Given a specification $\varphi$ and a planned trajectory from the motion planning module of the ADS, \coolname first determines whether the trajectory is likely to violate $\varphi$.
To achieve this, \coolname first constructs a trace $\pi$ from the planned trajectory, i.e., by evaluating all variables and functions relevant to $\varphi$ at every time point with respect to the planned trajectory and the predicted environment.
For instance, given the planned trajectory in Table~\ref{tab:example_of_the_planned_trajectory} (and the predicted environment in Table~\ref{tab:example_of_the_predicted_environment}),
Table~\ref{tab:example_of_signals} shows the constructed trace.

One practical complication is that some variables relevant to $\varphi$ cannot be obtained from the planned trajectory as they are only known after command generation (see Section~\ref{sec:control_evaluation}).
For example, the values of $\mathtt{fogLight}$ (i.e., whether the fog light is on) 
and $\mathtt{warningFlash}$ can only be determined once the respective commands are generated.
For such situations, we use typed `placeholder' variables $\mathtt{x}_{i,j}$ in the scenes of the trace for each time step $i$ and position $j$.
We define an \emph{assignment} $\alpha$ to be a function mapping the typed variables $\mathtt{x}_{i,j}$ to the value domains.
Then, for traces $\pi$ containing those variables, $\pi$ satisfies $\varphi$ if and only if there exists an assignment $\alpha$ such that $\pi[\alpha(\mathtt{x}_{i,j})/\mathtt{x}_{i,j}]$ satisfies $\varphi$ for every variable $\mathtt{x}_{i,j}$ in $\pi$. Practically, finding a suitable assignment $\alpha$ is straightforward: all variables for assignment $\alpha$ have only a few possible discrete values (e.g., the light is on or off), and thus brute force search is sufficient and inexpensive.

Next, \coolname computes how `close' the ego vehicle will come to violating $\varphi$. Note that our goal is to proactively react when a violation is likely in the near future. This is because the ego vehicle operates in an open environment (e.g., with other vehicles and pedestrians) and thus reacting too late may be too risky if the predicted environment turns out to be wrong (e.g., a sudden move of a pedestrian).  
To measure how close a trace $\pi$ is to violating $\varphi$, we adopt a quantitative semantics~\cite{maler2004monitoring, deshmukh2017robust,nivckovic2020rtamt} that produces a numerical \emph{robustness} degree. 

\begin{definition}[Quantitative Semantics]\label{def:Quantitative_Semantics}
Given a trace $\pi$ and a formula $\varphi$, the quantitative semantics is defined as the robustness degree $\rho(\varphi, \pi,t)$, computed as follows.
Recall that propositions $\mu$ are of the form $f(x_0,x_1,\cdots,x_k) \sim 0$.
\begin{equation*}
  \rho(\mu, \pi, t) =
    \begin{cases}
      -\pi_t(f(x_0,x_1,\cdots, x_k)) & \text{if $\sim$ is $\leq$ or $<$}\\
      \pi_t(f(x_0,x_1,\cdots, x_k)) & \text{if $\sim$ is $\geq$ or $>$}\\
      \mid \pi_t(f(x_0,x_1,\cdots, x_k)) \mid & \text{if $\sim$ is $\neq$}\\
      -\mid \pi_t(f(x_0,x_1,\cdots, x_k)) \mid & \text{if $\sim$ is $=$}
    \end{cases}       
\end{equation*}
where $t$ is the time step and $\pi_t(e)$ is the valuation of expression $e$ at time $t$ in $\pi$.
\begin{align*}
\rho(\neg\varphi,\pi,t) & = -\rho(\varphi,\pi,t) \\
\rho(\varphi_1 \land \varphi_2,\pi,t) & = \min\{\rho(\varphi_1,\pi,t),\rho(\varphi_2,\pi,t)\} \\
\rho(\varphi_1 \lor \varphi_2,\pi,t) & = \max\{\rho(\varphi_1,\pi,t),\rho(\varphi_2,\pi,t)\} \\
\rho(\varphi_1 \;\mathtt{U_I}\; \varphi_2,\pi,t) & = \sup_{t_1 \in t+\mathtt{I}} \min \{\rho(\varphi_2,\pi,t_1), \inf_{t_2 \in [t,t_1]} \rho(\varphi_1,\pi,t_2)\}
\end{align*}
where $t+I$ is the interval $[l+t,u+t]$ given $I=[l,u]$.
\qed
\end{definition}
Note that the smaller $\rho(\varphi,\pi,t)$ is, the closer $\pi$ is to violating $\varphi$.
If $\rho(\varphi,\pi,t) \leq 0$, $\varphi$ is violated.
We write $\rho(\varphi,\pi)$ to denote $\rho(\varphi,\pi,0)$; $\pi \vDash \varphi$ to denote $\rho(\varphi,\pi,t) > 0$; and $\pi \not \vDash \varphi$ to denote $\rho(\varphi,\pi,t) \leq 0$. Note that time is discrete in our setting. 

\begin{table}[t]
\setlength{\abovecaptionskip}{0.cm}
	\centering
	\caption{Trace obtained from the trajectory in Table~\ref{tab:example_of_the_planned_trajectory}}
	\label{tab:example_of_signals}

    \begin{tabular}{|l|l|c|c|c|c|c|} \hline
         \multicolumn{2}{|c|}{ planning signals} & 0 & 2 & 4 & 6 & 8  \\
         \hline

         \multicolumn{2}{|c|}{$\mathtt{speed}$} & 7.01 & 6.13 & 5.44 & 5.09 & 3.89 \\
         \multicolumn{2}{|c|}{$\mathtt{direction}$} & 0 & 0 & 0 & 0 & 0 \\

         \multicolumn{2}{|c|}{$\mathtt{D(stopline)}$} & 44 & 30.66 & 19.17 & 8.15 & -0.75 \\
         \multicolumn{2}{|c|}{$\mathtt{D(junction)}$} & 44 & 30.66 & 19.17 & 8.15 & -0.75 \\

         \multicolumn{2}{|c|}{$\mathtt{fogLight}$} & $\mathtt{x}_{0,0}$ & $\mathtt{x}_{2,0}$ & $\mathtt{x}_{4,0}$ & $\mathtt{x}_{6,0}$ & $\mathtt{x}_{8,0}$ \\ 

         \multicolumn{2}{|c|}{$\mathtt{warningFlash}$} & $\mathtt{x}_{0,1}$ & $\mathtt{x}_{2,1}$ & $\mathtt{x}_{4,1}$ & $\mathtt{x}_{6,1}$ & $\mathtt{x}_{8,1}$ \\ \hline

         \multicolumn{2}{|c|}{Prediction Signals} & 0 & 2 & 4 & 6 &  8\\
         \hline

         \multicolumn{2}{|c|}{$\mathtt{TL(color)}$} & 1 & 0 & 0 & 0 & 2 \\

         \multicolumn{2}{|c|}{$\mathtt{fog}$} & 0.6 & 0.6 & 0.6 & 0.6 & 0.6 \\
         \multicolumn{2}{|c|}{$\mathtt{PriorityV(20)}$} & $\mathtt{false}$ & $\mathtt{false}$ & $\mathtt{false}$ & $\mathtt{false}$ & $\mathtt{false}$ \\
         \multicolumn{2}{|c|}{$\mathtt{PriorityP(10)}$} & $\mathtt{false}$ & $\mathtt{false}$ & $\mathtt{false}$ & $\mathtt{true}$ & $\mathtt{true}$ \\ \hline
    \end{tabular}
\vspace{-0.5cm}
\end{table}

\begin{example}
\label{example:robustness calculation}
Let $\varphi = \Box (speed < 90)$, i.e., the speed limit is $90$km/h.
Suppose $\pi$ is $\langle (speed \mapsto 0, \dots), (speed \mapsto 0.5, \dots), \cdots \\ (speed \mapsto 85, \dots) \rangle$ where the ego vehicle's max $speed$ is $85$km/h at the last time step.
We have $\rho(\varphi, \pi) = \rho(\varphi, \pi, 0) = min_{t \in [0, |\pi|]} ( 90 - \pi_t(speed) ) = 5$.
Suppose instead that $\varphi$ is the specification from Example~\ref{example:specifications} and $\pi$ is the trace from Table~\ref{tab:example_of_signals}. The robustness value is $\rho(\varphi,\pi) = 0$, i.e., $\varphi$ is violated as the ego vehicle fails to stop before the stop line when the traffic light turns red.
\qed
\end{example}

\subsection{Trajectory Repair}
\label{sec:PlanRepair}

If the robustness value of $\varphi$ with respect to a trace $\pi$ is below a certain threshold $\theta$, there is a risk of violating $\varphi$ in the future, even if $0 < \rho(\varphi,\pi) \leq \theta$ (given that there is uncertainty in the predicted environment).
This threshold is determined experimentally in our work (Section~\ref{sec:implementation_evaluation}): intuitively, it characterises how `cautious' the ADS is.
In order to enforce $\varphi$, i.e., proactively prevent its possible violation, \coolname \emph{repairs} the planned trajectory before sending it to the control module of the ADS so as to change the commands that will be generated.

Our trajectory repair method consists of three steps.
First, we identify the earliest time step when the robustness value falls below the threshold.
Second, we compute the gradient (through auto-differentiation~\cite{griewank1989automatic}) of each variable at the identified time step with respect to the robustness degree.
Based on the result, we then modify the variable to increase the robustness degree.
Finally, we modify the planned trajectory accordingly and feed it into the control module. In the following, we present each step in detail. \\

\noindent \textbf{\emph{Determine the time step.}}
Given a trace $\pi = \langle (t_0, s_0), \cdots, (t_n, s_n) \rangle$, we write $\pi^k$ to denote the prefix $ \langle (t_0, s_0), (t_1, s_1), \cdots, (t_k, s_k) \rangle $. 
Given $\pi$ such that $\rho(\varphi, \pi) < \theta$, we aim to identify a time step $k$ such that: (1) $\rho(\varphi, \pi^k) < \theta$; and (2) there does not exist a time step $l$ such that $l < k$ and $\rho(\varphi, \pi^l) < \theta$. 
Intuitively, $k$ is the earliest time step when the robustness value falls below the threshold. 
We identify the time step $k$ using a sequential search, i.e., we start from $k=0$ and keep increasing $k$ until we find a $k$ such that $\rho(\varphi, \pi^k) < \theta$.

\begin{example}
\label{example:timestep location}
Let $\varphi = law38_{3}$ from Example~\ref{example:specifications} and $\pi$ denote the trace from Table~\ref{tab:example_of_signals}.
Suppose the threshold $\theta$ is 10.
Then, as shown in Example~\ref{example:robustness calculation}, $\rho(law38_{3}, \pi) = 0$ and is thus below the threshold. Then, we apply the above-mentioned algorithm to identify the time step. 
The following are computed in sequence. 
\begin{align*}
    \rho(\varphi, \pi^{0}) = 42, \cdots, \rho(\varphi, \pi^2)  = 28.66, \cdots,\\
    \rho(\varphi, \pi^{4}) = 17.17, \cdots, \rho(\varphi, \pi^{6}) = 6.15 
\end{align*}
Thus, the time step $k$ that we are looking for is $6$ (as $6 < \theta$). 
\qed
\end{example}

\noindent \textbf{\emph{Calculate the gradient.}}
Next, we find out how the variables at time step $k$ should be modified so that the robustness degree of the resulting trace can be improved.
We thus define a differentiation function that calculates the gradient of each relevant variable with respect to the robustness degree.
Intuitively, when the gradient of a variable $x$ at time $k$ is positive (resp.~negative), we can increase the robustness degree by increasing (resp.~decreasing) the value of $x$. 

Recall that the robustness degree of $\varphi$ is computed using discrete functions $min$ and $max$ (Definition~\ref{def:Quantitative_Semantics}) that are hard to differentiate~\cite{wen2012using}.
Hence, we adopt a continuous robustness measure as defined in~\cite{pant2017smooth, gilpin2020smooth} which replaces $min$ and $max$ in Definition~\ref{def:Quantitative_Semantics} with continuous functions $\widetilde{max}$ and $\widetilde{min}$ as follows:
\begin{align*}
    \widetilde{max}\{x_0, x_1, \dots, x_m\} & = \frac{1}{a} \ln ( \sum_{i=1}^{m} e^{a x_i}) \\
    \widetilde{min}\{x_0, x_1, \dots, x_m\} & = - \widetilde{max}(-x_0, -x_1, \dots, -x_m)
\end{align*}
where $a$ is a constant that controls the accuracy of $\widetilde{max}$ and $\widetilde{min}$. The larger $a$ is, the closer $\widetilde{max}$ (resp.~$\widetilde{min}$) is to $max$ (resp.~$min$).
We set $a$ to be 10, following~\cite{pant2017smooth, gilpin2020smooth}.
We denote the continuous robustness degree as $\tilde{\rho}(\varphi, \pi)$. The following proposition from~\cite{pant2017smooth, gilpin2020smooth} establishes the soundness of approximating $\rho(\varphi, \pi)$ with $\tilde{\rho}(\varphi, \pi)$. 
\begin{proposition}
Let $\varphi$ be an STL formula, $\pi$ be a trace, and $\varepsilon$ be a real value larger than 0. Then, there exists a value $a_1$ such that $|\Tilde{\rho}(\varphi, \pi, i) - \rho(\varphi, \pi, i)| < \varepsilon$ holds for all $a > a_1$.
\hfill \qed
\end{proposition}

Next, we define a differentiation function $D(\varphi, \pi, x^k)$ that returns a given variable $x$'s gradient with respect to $\rho (\varphi, \pi)$ at time $k$.
\begin{align*}
D(\varphi, \pi, x^k) = \frac{\partial \tilde{\rho}(\varphi, \pi, 0)}{\partial x^k}
\end{align*}
The following shows how $D(\varphi, \pi, x^k)$ is computed.

\begin{definition} 
\label{def:Gradient}
Given an STL formula $\varphi$ and trace $\pi$, function $D(\varphi, \pi, x^k)$ is defined as follows:
\begin{equation*}
  \frac{\partial \tilde{\rho}\left( \mu, \pi, t\right)}{\partial x^k} =
    \begin{cases}
      0 & \text{if $k \neq t$}\\
      \frac{d f'(x_0,x_1,\cdots,x_n)}{d  x^k} & \text{otherwise}\\
    \end{cases}       
\end{equation*}
where $\frac{d f'(x_0,x_1,\cdots,x_k)}{d x^k}$ is the derivative of function $f'$ with respect to $x^k$. Furthermore, let $\frac{\partial \widetilde{max} (\{x_0, x_1, \dots, x_m \}) }{\partial x^k}$ be defined as $\frac{e^{a x}}{\sum_{i=1}^{m} e^{a x_i}}$, and 
    let $\frac{\partial \widetilde{min} (\{x_0, x_1, \dots, x_m\})}{\partial x^k}$ be defined as $\frac{e^{-a x}}{\sum_{i=1}^{m} e^{-a x_i}}$.
{\small
\begin{align*}
\frac{\partial \tilde{\rho}\left( \lnot \varphi, \pi, t\right)}{\partial x^k} & = 
    -\frac{\partial \tilde{\rho}\left(\varphi, \pi, t\right)}{\partial x^k}\\
\frac{\partial \tilde{\rho}\left( \varphi_1 \land \varphi_2, \pi, t\right)}{\partial x^k} & = 
    \frac{\partial \widetilde{min} \{  \tilde{\rho}\left(\varphi_1, \pi, t \right) , \tilde{\rho}\left(\varphi_2, \pi, t \right)  \} }{\partial \tilde{\rho}\left( \varphi_1, \pi, t \right)} \cdot \frac{\partial \tilde{\rho}\left(\varphi_1, \pi, t \right)}{\partial x^k} \\
    &+ \frac{\partial \widetilde{min} \{  \tilde{\rho}\left(\varphi_1, \pi, t \right) , \tilde{\rho}\left(\varphi_2, \pi, t \right)  \} }{\partial \tilde{\rho}\left( \varphi_2, \pi, t\right)} \cdot \frac{\partial \tilde{\rho}\left(\varphi_2, \pi, t\right)}{\partial x^k}\\
\frac{\partial \tilde{\rho}\left( \varphi_1 \lor \varphi_2, \pi, t\right)}{\partial x^k} & = 
    \frac{\partial \widetilde{max} \{  \tilde{\rho}\left(\varphi_1, \pi, t \right) , \tilde{\rho}\left(\varphi_2, \pi, t \right)  \} }{\partial \tilde{\rho}\left( \varphi_1, \pi, t \right)} \cdot \frac{\partial \tilde{\rho}\left(\varphi_1, \pi, t \right)}{\partial x^k} \\
    &+ \frac{\partial \widetilde{max} \{  \tilde{\rho}\left(\varphi_1, \pi, t \right) , \tilde{\rho}\left(\varphi_2, \pi, t \right)  \} }{\partial \tilde{\rho}\left( \varphi_2, \pi, t\right)} \cdot \frac{\partial \tilde{\rho}\left(\varphi_2, \pi, t\right)}{\partial x^k}\\
\frac{\partial \tilde{\rho}\left( \varphi_1 \;\mathtt{U_I}\; \varphi_2, \pi, t\right)}{\partial x^k} & = 
    \sum_{t'\in t+\mathtt{I}} 
    \left( 
    \frac{\partial \tilde{\rho}\left( \varphi_1 \;\mathtt{U_I}\; \varphi_2, \pi, t \right)}{\partial  \tilde{\rho}\left( \varphi_1, \pi, t' \right)} 
    \cdot \frac{\partial \tilde{\rho}\left( \varphi_1, \pi, t'\right)}{\partial  x^k} 
    \right. \\
    & \left. + \frac{\partial \tilde{\rho}\left( \varphi_1 \;\mathtt{U_I}\; \varphi_2, \pi, t \right)}{\partial  \tilde{\rho}\left( \varphi_2, \pi, t' \right)} 
    \cdot \frac{\partial \tilde{\rho}\left( \varphi_2, \pi, t'\right)}{\partial  x^k} \right)
\end{align*}
}
where $\frac{\partial \tilde{\rho}\left( \varphi_1 \;\mathtt{U_I}\; \varphi_2, \pi, t \right)}{\partial  \tilde{\rho}\left( \varphi_1, \pi, t' \right)}$ is the derivative of $\tilde{\rho}\left( \varphi_1 \mathtt{U_I} \varphi_2, \pi, t \right)$ with respect to $\tilde{\rho}\left( \varphi_1, \pi, t' \right)$, and is defined as:
{\tiny
\begin{align*}    
    \sum_{t_1\in t+\mathtt{I} \land t_1 \geq t'} 
    \left( \frac{\partial \widetilde{max} \{ \widetilde{min} \{ \tilde{\rho}( \varphi_2,\pi,t_1), \inf_{t_2 \in [t,t_1]} \tilde{\rho}(\varphi_1,\pi,t_2)\} ~| t_1\in t+\mathtt{I} \} }{\partial \widetilde{min} \{\tilde{\rho}(\varphi_2,\pi,t_1), \inf_{t_2 \in [t,t_1]} \tilde{\rho}(\varphi_1, \pi,t_2)\} } \cdot   \right.  \\
     \left. \frac{\partial \widetilde{min} \{  \tilde{\rho}(\varphi_2,\pi,t_1), \inf_{t_2 \in [t,t_1]} \tilde{\rho}(\varphi_1,\pi,t_2) \} }{\partial \inf_{t_2 \in [t,t_1]} \tilde{\rho}(\varphi_1,\pi,t_2)  }
    \cdot 
    \frac{\partial \widetilde{min} \{\tilde{\rho}(\varphi_1,\pi,t_2) ~|~ t_2 \in [t,t_1]\} }{\partial \tilde{\rho}(\varphi_1,\pi,t')  } \right)
\end{align*}
}
where $\frac{\partial \tilde{\rho}\left( \varphi_1 \;\mathtt{U_I}\; \varphi_2, \pi, t \right)}{\partial  \tilde{\rho}\left( \varphi_2, \pi, t' \right)}$ is defined as:
{\small
\begin{align*}
    \frac{\partial \widetilde{max} \{ \widetilde{min} \{\tilde{\rho}(\varphi_2,\pi,t_1), \inf_{t_2 \in [t,t_1]} \tilde{\rho}(\varphi_1,\pi,t_2)\} ~|~ t_1\in t+\mathtt{I} \} }{\partial \widetilde{min} \{\tilde{\rho}(\varphi_2,\pi,t'), \inf_{t_2 \in [t,t']} \tilde{\rho}(\varphi_1,\pi,t_2)\} } \cdot  &\\
     \frac{\partial \widetilde{min} \{  \tilde{\rho}(\varphi_2,\pi,t'), \inf_{t_2 \in [t,t']} \tilde{\rho}(\varphi_1,\pi,t_2) \} }{\partial \tilde{\rho}(\varphi_2,\pi,t')  } &
\end{align*}  \qed
}
   
\end{definition}

Given the time step $k$ previously identified, we apply the above definition to compute $D(\varphi, \pi^k, x^k)$ for every variable $x$.
The purpose of the differentiation function D is to determine the “responsibility” of each signal in violating the specification. In other words, consider the computation of robustness as a function of multiple variables where D determines the gradient of each variable.
We remark that our implementation of $D(\varphi, \pi^k, x^k)$ is based on automatic differentiation techniques~\cite{griewank1989automatic}.
Intuitively, we store the intermediate values while computing the robustness degree, and then compute the gradients based on reverse accumulation.

\begin{example}
\label{example:gradient calculation}
Given the trace $\pi$ of Table~\ref{tab:example_of_signals}, the following shows how to calculate the gradient of $speed$ with respect to $\varphi_0 = \Box (speed > 5)$ at time step $6$:
\small{
\begin{align*}
    & D(\varphi_0, \pi^6, speed^6) = \frac{\partial \rho(\varphi, \pi^6, 0)}{\partial \rho(speed > 5, \pi^6, 6)} \cdot \frac{\partial \rho(speed > 5, \pi^6, 6)}{\partial speed^6} \\
    & = \frac{e^{-10 * 0.09}}{e^{-10 * 2.01} + e^{-10 * 1.13} + e^{-10 * 0.44} + e^{-10 * 0.09}} \cdot 1 = 0.97
\end{align*}
}
Similarly, continuing Example~\ref{example:timestep location}, the gradients are computed as follows: 
{\small
\begin{align*}
D(law38_{3}, \pi^6, speed^6) &=  8.39\times10^{-08} \\
D(law38_{3}, \pi^6, D(stopline)^6) &=  0.5 \\
D(law38_{3}, \pi^6, D(junction)^6) &= 0.5 \\
D(law38_{3}, \pi^6, direction^6) &= -4.74\times10^{-19} \\
D(law38_{3}, \pi^6, TL(color)^6) &= -9.48\times10^{-19} \\
D(law38_{3}, \pi^6, PriorityV(20)^6) &= -9.77\times10^{-28} \\
D(law38_{3}, \pi^6, PriorityN(20)^6) &= 2.15\times10^{-23}
\end{align*}
}
The gradients for variable $D(stopline)$ and $D(junction)$ at time step 6 are positive, which means that we can effectively increase the robustness value $\rho ( law38_{3}, \pi^6)$ by increasing $D(stopline)^6$ or $D(junction)^6$.
\qed
\end{example}

\begin{proposition}
\label{prop:alwaysincrease}
    Let $D(\varphi, \pi, x^k)$ be the result of gradient calculation as shown in Definition~\ref{def:Gradient}. When $D(\varphi, \pi, x^k)$ is positive (or negative), there exists an interval $(0, \Delta)$ such that increasing (or decreasing) $x^k$ within this interval increases in the value of $\tilde{\rho}(\varphi, \pi)$.
\end{proposition}

\begin{proof}
First, if $\varphi$ is a Boolean Expression $\mu$, $\tilde \rho(\mu, \pi)$ can be represented by a continuous function $f'(x_0,\cdots,x_n)$ as shown in Definition~\ref{def:Quantitative_Semantics}. Given that $f'(x_0,\cdots,x_n)$ is confined to linear or absolute value functions, the proposition holds for $\mu$.

Then, assuming the proposition holds, the proposition holds if we can prove the proposition holds for each and every way $\varphi$ can be constructed, i.e., $\lnot \varphi_1$, $\varphi_1 \land \varphi_2$, $\varphi_1 \lor \varphi_2$, and $\varphi_1 \;\mathtt{U_I}\; \varphi_2$.

If $\varphi$ is in the format of $\lnot \varphi_1$, we have
$
\tilde\rho(\lnot \varphi_1, \pi) = - \tilde \rho(\varphi_1, \pi), 
$
 and
$
D(\lnot \varphi_1, \pi, x^k) = -D(\varphi_1, \pi, x^k).
$ 
Thus, by negating the modification, we can ensure that the proposition holds for $\lnot \varphi_1$.

If $\varphi$ is in the format of $\varphi_1 \lor \varphi_2$, 
$
\tilde \rho(\varphi_1 \lor \varphi_2, \pi) = \frac{1}{a} \ln ( e^{a x_1} + e^{a x_2})
$, and
$D(\varphi_1 \lor \varphi_2, \pi, x^k) =
\frac{e^{a x_1}}{e^{a x_1} + e^{a x_2}} \cdot D(\varphi_1, \pi, x^k) +  
\frac{e^{a x_2}}{e^{a x_1} + e^{a x_2}} \cdot D(\varphi_2, \pi, x^k)
$. 
Here, $x_1 = \tilde \rho(\varphi_1, \pi)$, $x_2 = \tilde \rho(\varphi_2, \pi)$, $a \to \infty$. 
Suppose the proposition holds for $\tilde \rho(\varphi_1, \pi)$ within interval $(0, \Delta_1)$, and holds for $\tilde \rho(\varphi_2, \pi)$ within interval $(0, \Delta_2)$.
If $x_1 \neq x_2$, suppose $x_1 > x_2$, then we have $\frac{e^{a x_1}}{e^{a x_1} + e^{a x_2}} \to 1$, $\frac{e^{a x_2}}{e^{a x_1} + e^{a x_2}} \to 0$, and $D(\varphi_1 \lor \varphi_2, \pi, x^k) \to D(\varphi_1, \pi, x^k)$. The proposition holds for interval $(0, \Delta_1)$. 
If $x_1 = x_2$, then $D(\varphi_1 \lor \varphi_2, \pi, x^k) > 0$ indicates 
$D(\varphi_1, \pi, x^k) + D(\varphi_2, \pi, x^k) > 0$. Even if one of $D(\varphi_1, \pi, x^k)$ and $D(\varphi_2, \pi, x^k)$ is negative, the value of $e^{a x_1} + e^{a x_2}$ still increases, leading to the increase of $\tilde \rho(\varphi_1 \lor \varphi_2, \pi)$. Let $\Delta'$ be a number larger than 0 and smaller than $min\{\Delta_1, \Delta_2\}$. The proposition holds for $(0, \Delta')$.

If $\varphi$ is in the format of $\varphi_1 \land \varphi_2$, since $\varphi_1 \land \varphi_2 =\lnot( \lnot \varphi_1 \lor \lnot \varphi_2)$,
we can deduce that the proposition always holds for $\varphi_1 \land \varphi_2$.

If $\varphi$ is in the format of $\varphi_1 \;\mathtt{U_I}\; \varphi_2$.
Since $\tilde \rho(\varphi_1 \;\mathtt{U_I}\; \varphi_2,\pi)$ is a combination of the function $\widetilde{max}$ and $\widetilde{min}$, we can deduce that the proposition always holds for $\varphi_1 \;\mathtt{U_I}\; \varphi_2$.

Therefore, we can conclude that the proposition holds.
\end{proof}
Intuitively, Proposition~\ref{prop:alwaysincrease} clarifies that the gradient calculation function $D(\varphi, \pi, x^k)$ reflects the changing trend of the robustness function $\tilde\rho(\varphi, \pi)$ in terms of variable $x^k$. 
However, the changing trend is sensitive to the variable's current value. If we increase the variable by too much, it may lead to a decrease in robustness. 
For instance, consider the specification: $\varphi = 10 < speed < 100$.
Suppose the current speed is 8, then the robustness is $-2$, and the gradient for speed $D(\varphi, \pi, speed)$ is 1, indicating that we should increase the value of speed. If we increase the speed within the interval $(0, 94)$, the robustness will always be larger than $-2$. However, if we increase the speed to 103, the robustness will become $-3$, resulting in a decrease.
Therefore, to guarantee an increase in robustness, it is necessary to limit the modification within an interval of $(0, \Delta)$.

\noindent \textbf{\emph{Repair the trajectory.}}
The gradients calculated above allow us to determine how to effectively increase the robustness degree. 
We can proceed to repair the trace by modifying the variable with the maximal absolute gradient at time step $k$. 
The \emph{magnitude} of the modification is calculated as follows: 
\begin{align*}
    \delta = \frac{\theta - \tilde\rho(\varphi, \pi^k)}{D(\varphi, \pi^k, x^k)};
    ~While \; \tilde\rho(\varphi, \pi') < \tilde\rho(\varphi, \pi)~Do\!:~ \left\{\delta \leftarrow \delta/2\right\}
\end{align*}
where $\pi'$ is the trace after the modification. This \emph{magnitude} of the modification indicates that we try to increase the robustness value to $\theta$. However, this adjustment might sometimes lead to overreactions, causing a decrease in the robustness value. In such cases, we reduce $\delta$ until we observe an increase in the robustness value, and the descent rate during this process follows a scale of $2^n$, enabling us to efficiently determine the magnitude.
For instance, according to Example~\ref{example:gradient calculation}, we should modify $D(stopline)$ or $D(junction)$ at time step 6 with a magnitude of $\frac{\theta - \rho(law38_{3}, \pi^{6})}{0.5} = 7.7$. This modification results in $\rho(law38_{3}, \pi^{6})$ increasing from 6.15 to 13.85.

\begin{proposition}
\label{prop:overallincrease}
    Let $x^k$ be a variable, and $\delta$ be the \emph{magnitude} of the modification on $x^k$. The robustness value $\tilde \rho(\varphi, \pi)$ always increases after the modification. 
\end{proposition}

\begin{proof}
The modification is triggered only when $\theta - \tilde\rho(\varphi, \pi) > 0$, which implies that $\delta$ and $D(\varphi, \pi, x^k)$ share the same sign.
As shown in Proposition~\ref{prop:alwaysincrease}, there exists an interval $(0, \Delta)$ in which the gradient value is effective. If the previous modification results in a decrease of $\tilde{\rho}(\varphi, \pi)$, we can ensure an increase in $\tilde{\rho}(\varphi, \pi)$ by decreasing $\lvert \delta \rvert$ to a value smaller than $\Delta$. The proposition holds.
\end{proof}

\begin{algorithm}[t]
\caption{Trajectory repair algorithm}\label{alg:runtime_fix_trajectory_algorithm}
\small
\KwIn{variable/function $x$, time step $k$, magnitude $\delta$}

    \uCase{$x$ is $speed$}{
        Set $speed^k$ to be $speed^k + \delta$\;
    }
    \uCase{$x$ is $direction$}{
        Choose a value $d_0$ (0, 1, or 2) that is closest to $direction^k + \delta$\;
        Set the $steer$ at time $k$ to 0 if $d_0 = 0$\;
        Otherwise set the $steer$ at time $k$ to 0.1 if $d_0 = 1$\;
        Otherwise set the $steer$ at time $k$ to -0.1 if $d_0 = 2$\;
    }
    \Case{$x$ is of the form $D(\_)$ or $Lane(\_)$}{
        Search for a coordinate $(a, b)$ (i.e., new position for the ego vehicle) such that $D(\_)$ becomes $D(\_)+\delta$ or $Lane(\_)$ becomes $Lane(\_) + \delta$\;
        Set the position of the ego vehicle at time $k$ to $(a, b)$\;
    }
\end{algorithm} 

Recall that our goal is to modify the planned trajectory so as to trigger different control commands.
While we may modify the trace arbitrarily, we cannot do the same for the planned trajectory.
First, some of the variables may not be controllable, e.g., the color of the traffic light is beyond the control of the ADS.
Second, a variable may have a specific domain of discrete values in the ADS (e.g., $direction$ has the value of 0, 1, or 2) and thus we can only choose one of those valid values. 
Finally, the value of a variable may be the result of a function which depends on the current and future scenes. 
For instance, $D(stopline)$ measures the distance from the ego vehicle (according to the planned trajectory) to the stop line ahead (according to the map).
In these situations, it is very difficult to translate the modification to the planned trajectory.
Thus, we focus on modifying those signals that the ADS has control over and modify the planned trajectory accordingly, which are $speed$, $acc$, $direction$, $Lane(\_)$ (i.e., which lane the ego vehicle should be in), and $D(\_)$ (i.e., how far the ego vehicle is from a certain artifact). These naturally correspond to what human drivers focus on.
Algorithm~\ref{alg:runtime_fix_trajectory_algorithm} describes how the planned trajectory is repaired with respect to a specific variable/function $x$, time step $k$, and magnitude $\delta$.
Note that the fixes are specific to certain variables since they may have specific domains.
In the case of $direction$, we are constrained to choose a value from ${0, 1, 2}$ and set the value in the trajectory accordingly. 
In the case of functions based on the ego vehicle's position (e.g., $D(stopline)$), we search for nearby coordinates that are close to the desired value while still remaining on the road. Note that the ADS's planning module and the control module are entirely black boxes to us. Therefore, we do not take into account the correlations between variables when modifying the planned trajectory. To do so would require the construction of an exhaustive physical model, essentially equivalent to rebuilding the planning module of the ADS.



\begin{example}
\label{example:mutation}
Given the planned trajectory in Table~\ref{tab:example_of_the_planned_trajectory}, 
and the repair computed for $D(stopline)$ and $D(junction)$ in Example~\ref{example:gradient calculation}. 
We modify the planned car position at time step 6 from \textbf{(0, 35.85)} to \textbf{(0, 28.15)} (so the vehicle should be positioned further from the junction), leaving the remaining planned trajectory unchanged. Note that changing the value of $speed$ can rectify the trajectory as well, however, the gradient values strongly suggest that optimizing the position of the waypoint is a more efficient approach.
\qed
\end{example}

\subsection{Runtime Enforcement}
\label{sec:control_evaluation}

\begin{algorithm}[t]
\caption{Runtime enforcement algorithm}\label{alg:runtime_enforcement_algorithm}
\small
\KwIn{specification $\varphi$, trajectory $\Gamma$, the threshold $\theta$}
Generate trace $\pi$ based on $\Gamma$\;
\If{$\rho(\varphi, \pi) \leq \theta$} 
{
    Identify the smallest $k$ such that $\rho(\varphi,\pi^k) \leq \theta$\;
    Compute $D(\varphi, \pi^k, x^k)$ for every controllable variable $x^k$\;
    Identify variable $x^k$ with the maximal absolute gradient\;
    Invoke Algorithm~\ref{alg:runtime_fix_trajectory_algorithm} to fix the trajectory\;
}
\end{algorithm}

We are now ready to present our runtime enforcement algorithm, as shown in Algorithm~\ref{alg:runtime_enforcement_algorithm}. 
First, we generate a trace $\pi$ based on the planned trajectory $\Gamma$ and check whether $\rho(\varphi, \pi) \leq \theta$.
If so, we proceed to identify the time step $k$ when the robustness degree falls below the threshold. 
Then we compute gradients for the variables at time $k$, identify the controllable one with the maximal absolute gradient (w.r.t.~the robustness degree) and repair the trajectory accordingly.
The repaired trajectory is then sent to the control module, which generates the commands accordingly (e.g., turn on/off beam, and apply brake).

Recall that the specification $\varphi$ may also constrain the generated commands, e.g., the need to signal before turning.
To make sure the commands generated do not violate $\varphi$, we introduce a \emph{control validation module} (refer to Figure~\ref{fig:workflow}) that intercepts and checks the generated commands, modifying them if necessary.
Recall that commands related to motion (e.g., brake, accelerate, steer, and gear) are generated according to the (repaired) trajectory and thus do not require modification.
We remark that these commands are mostly simple in nature (i.e., with Boolean values) and thus we can easily modify them according to the specification. 
For instance, consider the beam-related signals, namely $\mathtt{highBeam}$ and $\mathtt{lowBeam}$, which have \emph{on} and \emph{off} states. We can easily modify these states by switching the values in the control commands sent to the AV's chassis control.

\begin{example}
\label{example:control validation}
Consider the trace shown in Table~\ref{tab:example_of_signals}.
Recall that the signals $\mathtt{fogLight}$ and $\mathtt{warningFlash}$ are not part of the planned trajectory: in fact, the ADS turns these off by default, i.e., we initially have $\alpha(\mathtt{x}_{i,j}) = \mathtt{false}$~($\mathtt{x}_{i,j}$ is the placeholder variable as discussed in Section~\ref{sec:constrct_trace}). 
To satisfy the specification in Example~\ref{example:specifications}, \coolname sets $\alpha(\mathtt{x}_{i,j}) = \mathtt{true}$ for each $i,j$. To realize this, we activate the $\mathtt{fogLight}$ and $\mathtt{warningFlash}$ in the control commands.
\qed
\end{example}

\section{Implementation and Evaluation}
\label{sec:implementation_evaluation}
We implemented \coolname for Apollo 6.0 and 7.0~\cite{apollo60,apollo70}. The code is on our website~\cite{ourweb2}.
In particular, we built a bridge program that interprets Apollo's messages (in JSON format) and obtains the values of variables and functions used by the specification language.
Some of these values are obtained directly (e.g., $speed$ and $acc$), but some require complex processing.
For example, to get the value of variable $NPCAhead.speed$ at time $t$, we obtain the planned position of the ego vehicle at time $t$ from the planning module, and check every NPC vehicle's predicted trajectory from the prediction module to identify the one that is ahead of the ego vehicle at time $t$.
Our implementation relies on a third party component provided by LawBreaker~\cite{Sun-Poskitt-et_al22a}.
In particular, we utilise the tool's specification language and the corresponding verification algorithm.

We conducted experiments to answer the following Research Questions~(RQs): 
\begin{itemize}
    \item[\textbf{RQ1}:] Can \coolname be used to enforce non-trivial specifications?
    \item[\textbf{RQ2}:] How much overhead is there for runtime enforcement?
    \item[\textbf{RQ3}:] Does \coolname minimise the enforcement? 
\end{itemize}
RQ1 considers whether \coolname achieves its primary goal of being able to enforce complex specifications (i.e.,~beyond collision avoidance).
RQ2 and RQ3 consider whether \coolname implements its enforcement in a way that is practically reasonable.
The former focuses on the overhead of runtime enforcement, since AVs are expected to react quickly on the road.
The latter focuses on the magnitude of the repair to the original trajectory, i.e., the enforcement should take place only if necessary and should minimally alter the behaviour of ADS.

Our experiments were run in the high-fidelity LGSVL simulator~\cite{rong2020lgsvl}.
Due to randomness in the simulator (mostly due to concurrency), each experiment was executed 100 times and we report the averages. 
The threshold was determined in a preliminary experiment in which we ran Apollo multiple times for each scenario to get the range of the possible robustness values.
All experiments were obtained using two machines with 32GB of memory, an Intel i7-10700k CPU, and an RTX 2080Ti graphics card. 
The machines respectively use Linux (Ubuntu 20.04.5 LTS) and Windows (10 Pro).

\begin{table}[!t]
\setlength{\abovecaptionskip}{0.cm}
\footnotesize
\centering
\caption{Violations of Chinese traffic laws}
\label{tab:fix_condition_of_RED}
\begin{tabular}{|>{\centering}m{0.07\linewidth}|>{\centering}m{0.05\linewidth}|c|c|>{\centering}m{0.1\linewidth}|c|c|}
\hline
\multicolumn{2}{|c|}{  \multirow{2}{*}{\textbf{traffic laws}}}  & \multicolumn{2}{c|}{ \textbf{enforced?}} & \multirow{2}{*}{\textbf{improve}} & \multirow{2}{*}{\textbf{fail reason}} & \multirow{2}{*}{\textbf{content}}\\ \cline{3-4}
\multicolumn{2}{|c|}{}                      & 6.0 & 7.0 &  & &  \\ \hline
\multirow{3}{*}{Law38}  & sub1              & $\surd$  & $\surd$    &   \multirow{3}{*}{+55.83\%}  &     -            & green light\\ \cline{2-4} \cline{6-7}
                            & sub2          & $\surd$  & $\surd$    &     &     -            & yellow light\\ \cline{2-4} \cline{6-7}
                            & sub3          & $\surd$  & $\surd$    &     &     -           & red light\\ \hline
\multicolumn{2}{|c|}{Law44}                 & $\surd$  & $\surd$    &   +30.00\%  &     -     & lane change\\ \hline       
\multirow{2}{*}{Law46}  & sub2              & $\surd$  & $\surd$    &   +44.00\%  &     -    & speed limit\\ \cline{2-7} 
                            & sub3          & $\times$  & $\times$  &-          & Lack support             & speed limit\\ \hline
\multicolumn{2}{|c|}{Law47}                 & $\times$  & $\times$  &-          & Lack support           & overtake\\ \hline  
\multirow{3}{*}{Law51}      & sub3          & $\times$  & $\times$  &-          & Lack support          & traffic light\\ \cline{2-7} 
                            & sub4          & $\surd$  & $\surd$    &  \multirow{2}{*}{+35.00\%}  &      -     & traffic light\\  \cline{2-4} \cline{6-7}
                            & sub5          & $\surd$  & $\surd$    &    &      -    & traffic light\\ \hline  
\multirow{2}{*}{Law57}  & sub1              & $\times$  & $\times$  &-          & Lack support          & left turn signal\\ \cline{2-7} 
                        & sub2              & $\times$  & $\times$  &-          & Lack support           & right turn signal\\ \hline 
\multicolumn{2}{|c|}{Law58}                 & $\times$  & $\times$  &-          & Lack support          & warning signal\\ \hline  
\multicolumn{2}{|c|}{Law59}                 & $\times$  & $\times$  &-          & Lack support            & signals\\ \hline 
\end{tabular}
\vspace{-0.5cm}
\end{table}

\textbf{\emph{RQ1: Can \coolname be used to enforce non-trivial specifications?}}
To answer this question, we adopted the formalisation of traffic laws reported in~\cite{Sun-Poskitt-et_al22a} as our specification and evaluated whether \coolname can be applied so that the ADS follows them. We remark the traffic laws are rather complicated as they model 13 testable traffic laws with many sub-clauses. Furthermore, we use the benchmark of scenarios provided by~\cite{Sun-Poskitt-et_al22a} in which Apollo is known to violate the specification.
We replay these violation-inducing scenarios with \coolname enabled, and report in Table~\ref{tab:fix_condition_of_RED} whether our approach is able to prevent the violations from occurring.
Each `subX' in Table~\ref{tab:fix_condition_of_RED} represents sub-rules of a traffic laws. For example, Law38 pertains to traffic light regulations and has three sub-rules, each covering the yellow, green, and red lights, respectively. The `enforced?' column indicates whether the enforcement is successful.
The enforcement was considered to be successful if the average passing rate of the specification after the enforcement is more than 50\% across all the repetitions. 
Note that Apollo's passing rate is always below 50\% for the selected scenarios.
The \emph{improve} column in Table~\ref{tab:fix_condition_of_RED} reports the average improvement of \coolname over Apollo, 
i.e., the maximum enhancement achieved by \coolname across a threshold value spectrum ranging from 0.0 to 1.2, with intervals of 0.1.
The improvement is calculated by subtracting the pass rate of Apollo from the pass rate of \coolname.
Note that since the sub laws of $\mathtt{law38}$ and $\mathtt{law51}$ are closely related, they are evaluated together for \emph{avg improve}.
The only reason that we cannot enforce some of the failed laws is that some simulator support is currently lacking. For example, we generate a command to turn on $\mathtt{fogLight}$ to satisfy the $\mathtt{law58}$ as shown in Example~\ref{example:specifications} and LGSVL's car model
ignores the command since it currently does not support fog lights.
We mark \emph{Lack support} in the table to illustrate this. 
The detailed improvement for all thresholds is shown in Figure~\ref{fig:overall_improve}. In this figure, the x-axis denotes the threshold value ($\theta$), while the y-axis signifies the average percentage improvement of \coolname over Apollo. As shown in Figure~\ref{fig:overall_improve},
\coolname successfully enforced all cases where an enforcement is feasible. 

\begin{figure}
    \centering
    \includegraphics[width=0.45\textwidth]{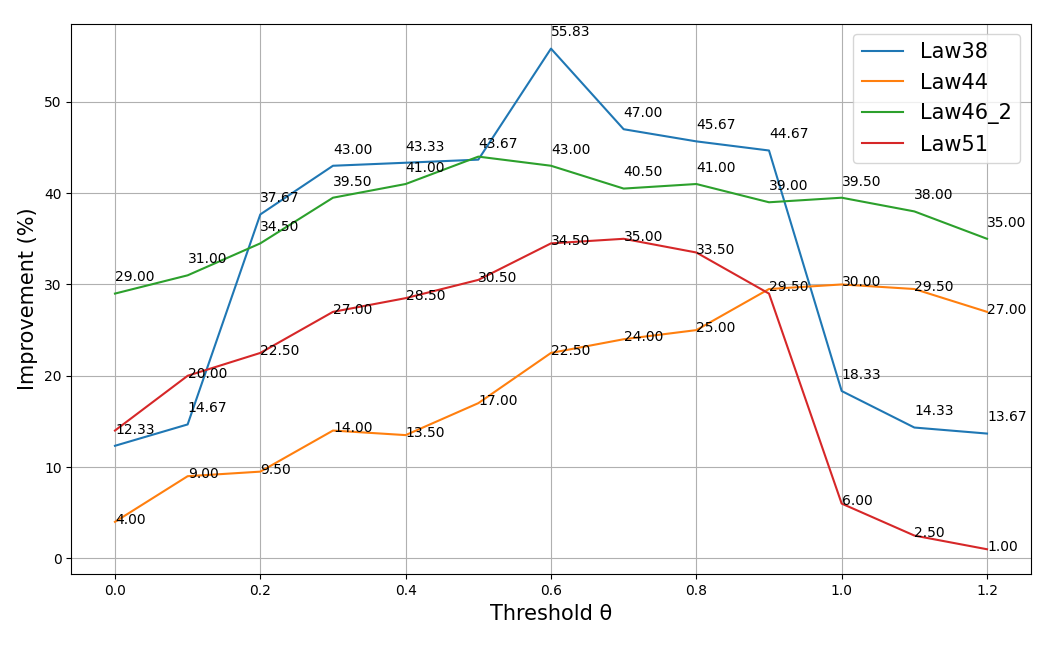} 
    \vspace{-0.5cm}
    \caption{Improvement of performance across thresholds}
    \label{fig:overall_improve}
    \vspace{-0.5cm}
\end{figure}

\begin{table}[t]
\setlength{\abovecaptionskip}{0.cm}
\centering
\caption{Performance comparison of \coolname and Apollo}
\label{tab:performance compare}
{\small
\begin{tabular}{>{\centering}m{0.12\linewidth} c >{\centering}m{0.03\linewidth} ccc >{\centering}m{0.03\linewidth} c}
\hline
Scenario  & Driver & $\theta$ & pass/total  & robustness  & avg time \\\hline

\multirow{14}{*}{\makecell[c]{Double \\Lined\\ Junction}}  
& Apollo7.0      & -    &   \textbf{40/100}         &  0.27      & 71.11s\\ 
& REDriver    & 0.0  &   \textbf{73/100}            & 0.67    & 55.18s\\ 
& REDriver    & 0.1  &   \textbf{78/100}            &   0.75 & 51.11s\\
& REDriver    & 0.2  &   \textbf{85/100}            &  0.96  & 50.56s\\ %

& REDriver    & 0.3  &   \textbf{93/100}            & 1.05    & 45.67s\\ 
& REDriver    & 0.4  &   \textbf{95/100}              & 1.07   & 45.71s\\
& REDriver    & 0.5  &   \textbf{93/100}              & 1.10   & 45.32s\\ %

& REDriver    & 0.6  &   \textbf{98/100}            & 1.19    & 44.92s\\ 
& REDriver    & 0.7  &   \textbf{99/100}            & 1.21  & 44.90s\\
& REDriver    & 0.8  &   \textbf{96/100}             & 1.15   & 45.13s\\ %
& REDriver    & 0.9  &   \textbf{99/100}    & 1.20    & 45.07s\\ 
& REDriver    & 1.0  &   \textbf{80/100}    &  0.78 & 51.90s\\
& REDriver    & 1.1  &   \textbf{75/100}    &  0.71  & 52.71s\\ %
& REDriver    & 1.2  &   \textbf{81/100}    & 0.80    & 53.10s\\ 
\hline
   
\multirow{14}{*}{\makecell[c]{Single \\Direction\\ Junction}} 
& Apollo7.0            & -          &   \textbf{18/100}   & 0.18    & 57.60s\\     
& REDriver          & 0.0           &   \textbf{26/100}   & 0.38    & 56.38s\\ 
& REDriver          & 0.1           &   \textbf{24/100}   & 0.37    & 55.93s\\
& REDriver          & 0.2           &   \textbf{85/100}   & 0.84    & 55.95s\\ %
& REDriver          & 0.3           &   \textbf{89/100}   & 0.86    & 55.75s\\ 
& REDriver          & 0.4           &   \textbf{89/100}   & 0.89    & 55.69s\\
& REDriver          & 0.5           &   \textbf{88/100}   & 0.87    & 56.71s\\ %
& REDriver          & 0.6           &   \textbf{91/100}   & 0.92    & 56.18s\\ 
& REDriver          & 0.7           &   \textbf{95/100}   & 0.99    & 56.13s\\
& REDriver          & 0.8           &   \textbf{93/100}   & 0.96   & 57.09s\\ %
& REDriver          & 0.9           &   \textbf{96/100}    & 1.02    & 57.55s\\  
& REDriver          & 1.0           &   \textbf{35/100}    & 0.56  & 57.10s\\
& REDriver          & 1.1           &   \textbf{31/100}    & 0.45   & 56.79s\\ %
& REDriver          & 1.2           &   \textbf{24/100}    & 0.33    & 56.61s\\ 
\hline
   
\multirow{14}{*}{T-Junction} 
& Apollo7.0         & -         &   \textbf{45/100}          &  0.29  & 64.66s\\ 
& REDriver          & 0.0       &   \textbf{41/100}    &  0.32  & 64.10s\\ 
& REDriver          & 0.1       &   \textbf{45/100}   &   0.43  & 63.93s\\
& REDriver          & 0.2       &   \textbf{46/100}   &  0.59   & 60.95s\\ %
& REDriver          & 0.3       &   \textbf{50/100}    & 0.77   & 61.82s\\ 
& REDriver          & 0.4       &   \textbf{49/100}    & 0.79  & 60.76s\\
& REDriver          & 0.5       &   \textbf{53/100}    & 0.45   & 62.30s\\ %
& REDriver          & 0.6       &   \textbf{55/100}    & 0.39   & 61.92s\\ 
& REDriver          & 0.7       &   \textbf{50/100}    &  0.41 & 61.70s\\
& REDriver          & 0.8       &   \textbf{51/100}    &  0.35  & 62.89s\\ %
& REDriver          & 0.9       &   \textbf{42/100}    & 0.33   & 63.23s\\ 
& REDriver          & 1.0       &   \textbf{43/100}    & 0.37  & 63.45s\\
& REDriver          & 1.1       &   \textbf{40/100}    & 0.40   & 63.70s\\ %
& REDriver          & 1.2       &   \textbf{39/100}    & 0.41    & 64.07s\\ 
\hline
\end{tabular}
}
\vspace{-0.5cm}
\end{table}

To explore RQ1 in more detail, we designed a second experiment that focused on $\mathtt{law38}$---one of the most complicated formulae in~\cite{Sun-Poskitt-et_al22a}---which specifies how a vehicle should behave at a traffic light junction (i.e., the constraints on movements due to green/yellow/red lights).
We then selected three scenarios highly relevant to this traffic law: \emph{Double Lined Junction}, \emph{Single Direction Junction}, and \emph{T Junction}.
The detailed specification $law38$ is given in our website~\cite{ourweb2}.
Note that these scenarios were generated by LawBreaker~\cite{Sun-Poskitt-et_al22a} to reliably induce traffic law violations in Apollo.
The seeds for the fuzzing algorithm are given on the website~\cite{ourweb2}.

We tested Apollo with and without \coolname on each violation-inducing scenario 100 times and recorded the pass rate and average robustness with respect to $\mathtt{law38}$.
Table~\ref{tab:performance compare} presents the result of our evaluation using Apollo 7.0 
(the results for version 6.0 are on our website~\cite{ourweb2}), where $\theta$ indicates threshold values.
Recall that the smaller the robustness value, the `closer' is a violation.

As can be seen from Table~\ref{tab:performance compare}, \coolname significantly outperforms the original Apollo in terms of respecting the specification.
In scenarios ``Double-Lined Junction'' and ``Single-Direction Junction'', the original Apollo failed to pass at the green light because it is too conservative at the junction. For instance, Apollo sometimes decides to stop before an intersection when there is enough space for the vehicle to pass safely.
\coolname avoided the violations by enforcing the vehicle to drive within the junction first or pass the junction directly. 
As a consequence, the average improvement to the pass rate is more than 50\% for scenario ``Double-Lined Junction'' and ``Single-Direction Junction''. 
For scenario ``T-Junction'', the improvement of \coolname is relatively small since there is heavy traffic in this scenario and Apollo sometimes produces the stop command.
By design, the stop command has higher priority than the planned trajectory since it prevents crashing in urgent situations.
Therefore, the enforcement did not take effect in some cases.

Furthermore, the performance of \coolname varies with threshold $\theta$, i.e., $\theta$ being too small or too large both lead to degraded performance.
If $\theta$ is too small, for example $0$ (which is equivalent to the driver being ignorant of what is going to happen), sometimes the ADS cannot enforce in time.
But if $\theta$ is too large (e.g., 1.0-1.2 which is equivalent to the driver being too scared of what might happen), \coolname is overcompensating and fixing things it should not, which can lead to unexpected ADS behaviour. 
Note that a significant number of valid trajectories have a robustness of $1$, and such an overreaction is more likely to occur for $\theta \geq 1$.

\begin{table}[t]
\setlength{\abovecaptionskip}{0.cm}
\centering
\caption{Overhead of \coolname}
\label{tab:overhead time compare}
{\small
\begin{tabular}{c  c c c c c c c}
\hline
 & $\theta$ & avg fix & max fix & fix (\%) & avg(ms) & max(ms) & time (\%)
\\\hline
\multirow{13}{*}{\makecell[c]{S1}}
     & 0.0       & 26.08 &  35   &  5.09\% & 1.92    & 6.82    & 4.88\% \\ 
      & 0.1       & 24.19 &   35  &  4.76\% & 1.88    & 9.11    & 4.81\% \\ 
     & 0.2       & 19.20  &  35   &  4.15\% & 1.90    & 9.10    & 4.72\% \\ 
     & 0.3       & 18.33 & 32    &  3.57\% & 1.87    & 9.23    & 5.05\% \\ 
    & 0.4       & 22.33 & 35    &  3.89\% & 1.85    & 9.86    & 4.98\% \\ 
      & 0.5       & 33.12 & 45    &  6.19\% & 1.91    & 8.81    & 4.90\% \\ 
     & 0.6       & 32.06 & 45    &  6.07\% &  1.72   & 9.08    & 4.89\% \\ 
     & 0.7       & 33.20 & 45    &  6.19\% & 1.88    & 9.12    & 5.04\% \\ 
     & 0.8       & 39.75 & 93    &  7.31\% & 1.95    & 9.10    & 4.92\% \\ 
     & 0.9       & 40.18 & 102   &  7.45\% &   1.96  & 10.53    & 4.85\% \\ 
     & 1.0      & 87.20 & 230    &  17.29\% & 2.13    & 11.15    & 6.35\% \\ 
     & 1.1      & 86.02 & 231    &  17.14\% & 2.05    & 10.55    & 6.01\% \\ 
    & 1.2       & 86.05 &  230  &  16.79\% &   2.11  & 11.13    & 6.26\% \\ 
\hline
    
\multirow{13}{*}{\makecell[c]{S2}} 
     & 0.0       & 10.22 & 17    &  2.51\% &   1.67  & 7.01    & 4.52\% \\ 
      & 0.1       & 10.71 & 17    &  2.55\% &   1.53  & 6.97    & 4.49\% \\ 
     & 0.2       & 10.56 & 20    &  2.78\% &   1.55  & 6.74    & 4.51\% \\ 
      & 0.3       & 12.05 & 22    &  2.93\% &   1.56  & 6.72    & 4.22\% \\ 
     & 0.4       & 12.10 & 22    &  2.90\% &   1.55  & 7.15    & 4.21\% \\ 
     & 0.5       & 12.21 & 20    &  2.95\% &   1.66  & 6.97    & 4.38\% \\ 
     & 0.6       & 12.23 & 20    &  2.98\% &   1.54  & 7.53    & 4.31\% \\ 
     & 0.7       & 12.48 & 22    &  2.73\% &   1.52  & 7.19    & 4.47\% \\ 
     & 0.8       & 14.35 & 22    &  3.11\% &   1.54  & 6.80    & 4.19\% \\ 
    & 0.9       & 14.75 & 20    &  3.60\% &   1.49  & 6.78    & 4.28\% \\ 
     & 1.0       & 170.15& 177    &  42.57\% &  1.88  & 9.15    & 5.89\% \\ 
     & 1.1       & 169.12& 185    & 40.82\% & 1.75    & 9.23    & 5.30\% \\ 
     & 1.2       & 169.20& 177   &  41.17\% &   2.09  & 9.09    & 5.61\% \\ 
\hline

\multirow{13}{*}{S3} 
     & 0.0       & 38.92 & 51    &  9.02\% &   1.63  & 9.14    & 4.18\% \\ 
     & 0.1       & 37.03 & 49    &  8.87\% &   1.67  & 9.21    & 4.50\% \\ 
     & 0.2       & 37.71 & 49    &  8.89\% &   1.70  & 9.45    & 4.39\% \\ 
     & 0.3       & 37.28 & 49    &  8.96\% &   1.84  & 9.21    & 4.54\% \\ 
     & 0.4       & 35.22 & 52    &  8.75\% &   1.82  & 9.50    & 4.43\% \\ 
     & 0.5       & 37.33 &  49   &  9.01\% &   1.71  & 9.34    & 4.70\% \\ 
     & 0.6       & 34.70 &   52  &  8.05\% &   1.69  & 10.13   & 4.90\%  \\ 
      & 0.7       & 33.13 & 52     &  7.80\% & 1.75    & 9.12    & 4.82\% \\ 
     & 0.8       & 35.56 &  40    &  8.35\% & 1.70    & 9.29    & 4.88\% \\ 
     & 0.9       & 32.46     &   40      &  7.76\% &   1.77  & 9.17    & 4.91\% \\ 
      & 1.0       & 233.73    & 298       &  52.05\% & 2.19    & 11.21    & 5.83\% \\ 
     & 1.1       & 232.76    & 298       &  52.18\% & 2.27    & 11.20    & 5.79\% \\ 
     & 1.2       & 234.46    &   298     &  51.99\% &   2.22  & 11.34    & 5.85\%  \\ 
\hline
\end{tabular}
}
\vspace{-0.5cm}
\end{table}

\textbf{\emph{RQ2: How much overhead does the runtime enforcement impose?}} 
To answer this question, we collect information on the running time of the plan validation module of \coolname for different scenarios. The overhead for control validation is very small~(less than 0.01\% of the the overhead of the plan validation module), and we ignore it in the later experiment.
The detailed data for \coolname based on Apollo 7.0 is shown in Table~\ref{tab:overhead time compare} (the results for version 6.0 are shown on our website~\cite{ourweb2}).
Here, \emph{S1-S3} corresponds to the \emph{Double Lined Junction}, \emph{Single-Direction Junction}, and \emph{T-Junction} as in Table~\ref{tab:performance compare}, \emph{avg fix} represents the average number of fixes during a test that successfully enables the ADS to follow the specification, \emph{max fix} represents the maximum number of fixes detected across the test cases, \emph{avg(ms)} means the average time consumption of the plan validation module in one run, \emph{max(ms)} indicates the maximum time consumption detected, \emph{fix ($\%$)} is calculated by dividing the average fixes by the average updates of the planned trajectory during a run, and \emph{time ($\%$)} is calculated by dividing the average time consumption of the plan validation module by the average time consumption of the production of a planned trajectory. The time units in Table~\ref{tab:overhead time compare} are all milliseconds.

As can be seen from Table~\ref{tab:overhead time compare}, the time consumption of the plan validation module is practical, i.e., the average time consumption is always smaller than 2.5 milliseconds, the max time consumption is always smaller than 12 milliseconds, and the \emph{time percent} is always within $6\%$.
Furthermore, the number of fixes is related to the value of $\theta$ as expected.
There is a large increase in the number of fixes and \emph{fix percent} for a large {$threshold \geq 1.0$ across all three scenarios.
This is consistent with the performance degradation at $threshold \geq 1.0$ shown in Table~\ref{tab:performance compare} since unnecessary fixes cause problems.

\textbf{\emph{RQ3: Does \coolname minimise the enforcement?}}
To answer this question, first, recall our approach as described in Section~\ref{sec:our_approach}.
We identify the smallest k with $\rho(\varphi, \pi^k) < \theta$. Hence, our fix applies only to the earliest part of the planned trajectory that leads to the near-violation of the specification.
In most cases, we only modify one variable at one time step, such as the ``speed'' at some time step. For some rare cases, we may modify multiple variables at one time step, such as the ``speed'' and ``position'', if modifying one single variable is not sufficient. Note that the ADS updates the planned trajectory based on current perceptions and predictions and the impact of our change does not accumulate.

\begin{figure}
    \centering
    \includegraphics[width=0.45\textwidth]{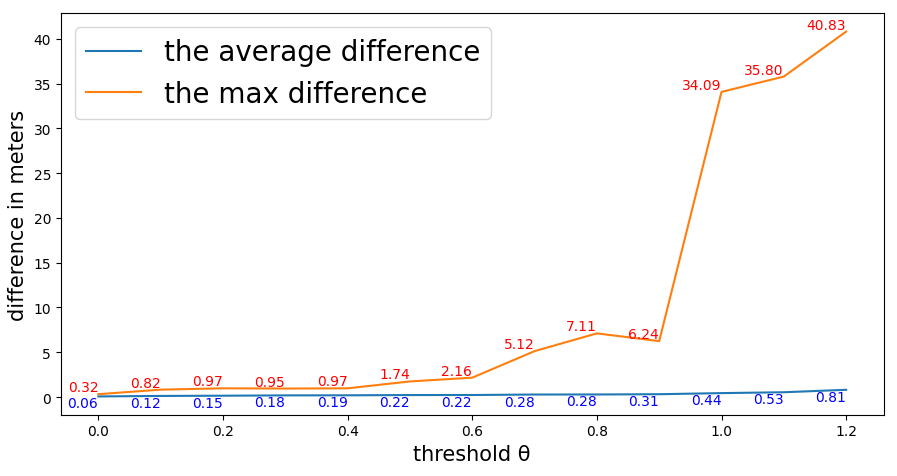} 
    \caption{Magnitude of modifications to planned trajectories}
    \label{fig:modification-diff}
\end{figure}

Here, we require a method to assess the variance between the modified planned trajectory and the original trajectory. This quantification is calculated by assessing the positional variance between these trajectories. When alterations are made to the speed or acceleration, we translate these changes into positional differences.  To be precise, the conversion for speed discrepancies is determined as follows: $(speed_t' - speed_t) \cdot I$, and for acceleration discrepancies: $(acc_t' - acc_t) \cdot I^2$, where $I$ signifies the time interval between the current planned \emph{waypoint} and the subsequent \emph{waypoint}. 
For instance, if a speed adjustment of magnitude $2m/s$ is applied to a planned waypoint, and the time interval $I$ is $0.2s$, then the positional difference is calculated as $2m/s \cdot 0.1s = 0.2 m$. The magnitude of modifications is shown in Figure~\ref{fig:modification-diff}. In this figure, the x-axis represents the threshold value $\theta$, and the y-axis represents the magnitude of the modification of \coolname in meters. 
The graph presented in the figure denotes the average/max modification value of \coolname across thousands of fixes. 
Notably, for thresholds ranging from $0.0$ to $1.2$, the average difference consistently remains below 1 meter. This observation suggests that \coolname's modification on the planned trajectory is small. Note that there is a significant increase in max difference for threshold $1.0$. This phenomenon is attributed to an excessive number of unnecessary fixes, as explained in RQ1.

In addition, the average running time for the test cases is listed in the last column of Table~\ref{tab:performance compare}. 
Here, \emph{avg time} represents the average time spent by the ADS to travel from the start point to the destination.
As can be seen from  the last column of Table~\ref{tab:performance compare}, the running time did not increase across all these test cases.
This indicates that \coolname did not, in practice, force the ADS to produce a substantially different trajectory to follow (e.g., halting the car).
Note in addition that the time consumption of \coolname has dropped substantially compared to the original Apollo in the scenario ``Double-Lined Junction''.
This is because Apollo hesitated at the green light, while \coolname successfully passed through.

\section{Related Work} \label{sec:related}

Runtime verification approaches monitor messages obtained from ADSs and evaluate them against a specification using a number of different techniques.
For instance, Kane et al.~\cite{Kane-et_al15a} generate a system trace from the observed network state, and Heffernan et al.~\cite{Heffernan-MacNamee-Fogarty14a} use system-on-a-chip based monitors as sources of information. 
Watanabe et al.~\cite{Watanabe-et_al18a} focus on runtime monitoring of the controller safety properties of advanced driver-assistance systems (ADASs).
Mauritz et al.~\cite{Mauritz-et_al16a} generate monitors for ADAS features from safety requirements and by training on simulators.
D'Angelo et al.~\cite{d2005lola} present Lola, a simple and expressive specification language to describe both correctness/failure assertions, which has been successfully deployed on autonomous vehicles in addition to many successful flight deployments.
Note that there is no enforcement of specifications in the works mentioned above. 

Runtime enforcement goes beyond monitoring and attempts to enforce certain safety properties. Existing works~\cite{hong2020avguardian, Cheng-et_al21a, Shankar-et_al20a, Grieser-et_al20a} already propose a few methods for runtime enforcement of ADSs. 
AVGuardian~\cite{hong2020avguardian} performs static analysis of the communication messages between the ADS modules to generate control policies and enforce them during runtime.
Guardauto et al.~\cite{Cheng-et_al21a} divide the ADS into a few partitions for the detection of rogue behaviours and restart the partition in order to clear them. Shankaro et al.~\cite{Shankar-et_al20a} define a policy using an automaton and enforce the car to stop when the policy is violated. Grieser et al.~\cite{Grieser-et_al20a} build an end-to-end neural network (from LIDAR to torques/steering) that implicitly picks up safety rules. Simultaneously, the distance to obstacles on the current trajectory is monitored and emergency brakes are applied if a collision is likely.
Generally, when enforcement for an ADS takes place in these works, it tends to be quite `weak', (e.g.~emergency brake).
\coolname, on the other hand, provides runtime enforcement for a rich specification in ways that are less intrusive.



In addition, there are existing works for cyber-physical systems~\cite{Pinisetty-et_al17a, Wu-et_al17a, Wu-et_al19a}. Pinisetty et al.~\cite{Pinisetty-et_al17a} formalise the runtime enforcement problem for CPSs, where policies depend not only on a controller but also an environment.
Another approach, Safety Guard~\cite{Wu-et_al17a}, adds automata-based reactive components to the original system, which react to ensure a predefined set of safety properties, while also keeping the deviation from the original system to a minimum. 
ModelPlex~\cite{mitsch2016modelplex} checks for model compliance of cyber-physical systems and includes a fail-safe action to avoid violations of safety properties. CBSA~\cite{phan2017component} proposes the idea of integrating assume-guarantee reasoning to allow runtime assurance of cyber-physical systems.  
These works are relevant to the runtime enforcement of ADSs since ADSs are cyber-physical systems as well. However, we can not directly apply these methods to ADSs and customization of the enforcement techniques is necessary due to the unique requirements and challenges posed by ADSs. For instance, the enforcement of ADSs requires consideration of not only the current control commands but also the planned trajectory.

Runtime enforcement is not limited to ADSs, i.e., there are works providing runtime enforcement/verification for general systems (e.g.,~\cite{sha2001using, bak2009system, Bloem-et_al15a, Schneider00a, Falcone-et_al12a, Ligatti-et_al09a, Falcone-et_al21a, desai2019soter, wongpiromsarn2011tulip, finucane2010ltlmop, schierman2015runtime, desai2017combining, copilot, schumann2015r2u2}).
The Simplex architecture~\cite{sha2001using, bak2009system} introduces the idea of ``runtime enforcement'' to enhance the reliability of complex software, and has been widely adopted in both academia and industry.
Shield synthesis~\cite{Bloem-et_al15a} proposes a method of runtime enforcement for reactive systems while also minimising interference to the original behaviour. Schneider~\cite{Schneider00a} looks at runtime enforcement of security policies and stops the program when they are violated. Falcone et al.~\cite{Falcone-et_al12a} propose enforcement by buffering actions and dumping them only when deemed safe. Ligatti et al.~\cite{Ligatti-et_al09a} use `edit automata' to respond to dangerous actions by suppressing them or inserting other actions. 
Desai et al.~\cite{desai2017combining} enforce the plan trajectory of mobile robots so as to follow STL specifications. Expanding upon this idea, Soter~\cite{desai2019soter} allows for safety properties to be specified and enforced in robotic systems. 
Tools such as TuLip~\cite{wongpiromsarn2011tulip} and LTLMoP~\cite{finucane2010ltlmop} synthesize trajectories to assist evaluation of the control system under linear temporal logic (LTL) specifications. Barron Associates provide a comprehensive study of runtime enforcement architecture for highly adaptive flight ontrol systems~\cite{schierman2015runtime}. 
The Copilot tool~\cite{copilot} offers a comprehensive runtime enforcement environment that incorporates numerous operating-system-like functionalities.
The R2U2~\cite{schumann2015r2u2} monitors the security properties of on-board Unmanned Aerial Systems (UAS) and is implemented in FPGA hardware. 
Unfortunately, many existing general runtime enforcement/verification methods are not suitable for ADSs due to their safety-critical and highly interactive nature. The survey paper by Falcone et al.~\cite{Falcone-et_al21a} on existing runtime enforcement/verification tools clarifies that many `reactions' provided by general runtime verification tools are weak, which is not acceptable for our situation. 
In this paper, we propose a runtime enforcement method applicable to any given specification with acceptable overhead for ADSs, and our method concerns not only the current driving conditions but also the ADS's future plans.

\section{Conclusion}
\label{sec:conclusion}
We proposed, \coolname, a solution to the runtime enforcement problem for ADSs. 
\coolname supports the enforcement of complex user-provided specifications such as national traffic laws in a way which is similar to experienced human drivers, i.e., based on near-future predictions and proactively correcting the vehicle's trajectory accordingly with minimal adjustment.

\section*{Acknowledgment}
We are grateful to the anonymous ICSE referees for their insights and feedback, which have helped to improve this paper.
This research is supported by the Ministry of Education, Singapore under its Academic Research Fund Tier 3 (Award ID: MOET32020-0004). Any opinions, findings and conclusions or recommendations expressed in this material are those of the author(s) and do not reflect the views of the Ministry of Education, Singapore.

\bibliographystyle{ACM-Reference-Format}
\bibliography{reference}

\appendix

\end{document}